\documentclass[smallextended]{svjour3}       
\usepackage{enumerate}
\usepackage{amsmath}
\usepackage{amssymb}
\usepackage{amstext}
\usepackage{comment, graphicx} 
\usepackage{multicol,graphicx,mathabx}

\journalname{Journal of Autonomous Agents and Multi-Agent Systems}
\begin{document}

\title{The Authorship Dilemma: Alphabetical or Contribution?}



\author{Margareta Ackerman 
\and
Simina Br\^anzei
}


\institute{Margareta Ackerman \at 
		San Jose State University\\
		Computer Science Department\\
		San Jose, CA\\
             \email{\texttt{margareta.ackerman@sjsu.edu}}           
           \and
             Simina Br\^anzei \at
		Hebrew University of Jerusalem\\
		Rachel and Selim Benin School of Computer Science and Engineering\\
		Jerusalem, Israel\\
		\email{\texttt{simina.branzei@gmail.com}}
}

\date{Received: date / Accepted: date}

\maketitle

\begin{abstract}
Scientific communities have adopted different conventions for ordering authors on publications. Are these choices inconsequential, or do they have  significant influence on individual authors, the quality of the projects completed, and research communities at large? What are the trade-offs of using one convention over another? 
In order to investigate these questions, we formulate a basic two-player game theoretic model, which already illustrates interesting phenomena that can occur in more realistic settings.

 We find that alphabetical ordering can improve research quality, while contribution-based ordering leads to a denser collaboration network and a greater number of publications.
Contrary to the assumption that free riding is a weakness of the alphabetical ordering scheme, this phenomenon can occur under any contribution scheme, and the worst case occurs under contribution-based ordering.
Finally, we show how authors working on multiple projects can cooperate to attain optimal research quality and eliminate free riding given either contribution scheme.

\end{abstract}

\section{Introduction}

Resource allocation is a central problem in artificial intelligence and more generally,
economic activity is fundamentally about resource allocation.
Without loss of generality,
every decision -- and thus every computation -- can be viewed as a resource allocation
instance (Wellman \cite{Wellman:EconAI}).
Research has been recognized as an economic
activity crucial for the long term well-being of society, and as a result, developed countries
allocate a significant percentage of their resources to basic and applied research activities.
The credit allocation schemes should incentivize scientific communities to operate at their best; however,
decisions regarding the distribution of research resources are sometimes made in an ad-hoc manner, with little theoretical or empirical justification of their long-term performance. In this paper, we investigate one of the core problems in this domain, namely, the allocation of credit for scientific work.

%
%
 
The allocation of scientific credit influences funding decisions, as well as tenure, promotions, and awards. Given the critical role that credit allocation plays in academia, it is surprising how little is known of the effects of name-ordering conventions. What influence do ordering schemes have on individual authors? Do they have any global effect on the research communities where they are applied and the kinds of projects completed therein? What can authors do to overcome the limitations imposed by credit allocation schemes? 

Our goal in this paper is to provide a framework for studying the effects of author ordering schemes and address these questions.
While many ordering conventions are possible, the prominent ones in academic communities
are to list authors \emph{by contribution}, that is, in descending order of their contribution to the paper, or \emph{alphabetically}, in lexicographical order of their last names. Studying the impact of these conventions will allow research communities to make informed decisions on whether to apply contribution or alphabetical author ordering schemes. 

Lake \cite{Lake} shows that listing authors alphabetically gives rise to the \emph{Matthew Effect}, whereby readers are likely to assume that the more established authors deserve more credit. 
Furthermore, alphabetical ordering can give an unfair benefit to those whose last names start with letters that occur earlier in the alphabet \cite{Efthyvoulou,Einav,Praag}. Tenure at highly ranked schools, fellowship, and to some extent even Nobel Prize winnings are correlated with surname initials~\cite{Einav}.
Indeed, the American Psychological Association~\cite{APA} mandates ordering authors by their contribution: 
\begin{quote}
\emph{``name of the principal contributor should appear first, with subsequent names in order of decreasing contribution.''}
\end{quote}

However, major disciplines such as mathematics, theoretical computer science, theoretical physics, and theoretical areas of economics have a long tradition of relying on alphabetical ordering. An empirical pattern about the attitude of researchers working in these areas is that authors start collaborations trusting others will do their best. Even if the contribution does not turn out proportional in the end, the authors are given equal credit, which is signaled through lexicographic ordering.
The American Mathematical Society~\cite{AMS} states:
\begin{quote}
\emph{``Determining which person contributed which ideas is often meaningless because the ideas grow from complex discussions among all partners... mathematicians traditionally list authors on joint papers in alphabetical order.''}
\end{quote}

More strikingly, there exist studies indicating that alphabetical ordering can result in improved research quality.
Laband and Tollison \cite{Laband} found that alphabetical ordering is positively correlated with research quality in economics research, where alphabetized papers are more highly cited than non-alphabetized papers. For example, the data from the American Economic Review---one of the big five journals in economics---shows that in the five years after publication, alphabetized articles are cited 50.16 times, while non-alphabetized articles are cited 30.38 times. From the means of the data in \cite{Laband}, this implies that alphabetized articles are cited 65\% more often than non-alphabetized articles. Brown, Chan, and Lai \cite{Brown} observe a similar trend in the marketing literature.

Joseph, Laband, and Patil \cite{Joseph} provide an extensive discussion of several studies on academic authorship, as well as synthetic data, by setting up a stochastic model of author ordering. \cite{Joseph} finds that the tendency towards alphabetical author ordering increases as acceptance rates decrease, and that for a fixed acceptance rate, papers whose authors are listed alphabetically tend to be of higher quality. 

While no scheme is optimal, it is important to first understand the strengths and weaknesses of each, such that each research community can implement the one most useful for its types of collaborations. Our investigation addresses projects completed by one or two parties. This is consistent with the rest of the literature on collaboration models and network formation games (where an edge exists between two researchers if they worked on a project together). One or two-authored papers represent a substantial fraction of the literature in some fields, including mathematics, physics, and economics. Laband and Tollison~\cite{Laband} show that, based on the number of citations, two authors appear to represent the optimal team size in economics. Newman~\cite{newman2004best} finds that:

\emph{``[...] purely theoretical papers appear to be typically the work of two scientists, with high-energy theory and computer science showing averages of 1.99 and 2.22."}

We collected data from the proceedings of major artificial intelligence conferences over the years 2013--2015 and found that approximately 30\% of the published papers are completed by one or two researchers (see Table 1). These numbers reinforce the 
study of one and two-authored papers as an important first step towards understanding the impact of credit allocation schemes.

\subsection{Our Contributions}

The main take-away message of our paper is the existence of counterintuitive effects due to author ordering schemes; crucially, quality is related to the effort the scheme can elicit from the participants. 
We make this point concrete by formulating
a game theoretic model of collaboration, which allows the investigation of credit allocation schemes and illustrates several important phenomena. 

Our formulation offers a compelling explanation for the phenomenon that alphabetical ordering can lead to improved research quality in some communities. In particular,
alphabetical ordering encourages collaborators to match each other's efforts:
 when one of the authors invests a lot of effort into a project, it is a best response for the co-author(s) to also invest high effort, since matching efforts leads to higher utility.

On the other hand, we find that contribution-based ordering can result in the completion of more research projects and a denser social network. The latter phenomenon has been observed through empirical data analysis by Newman \cite{Newman}.

When it comes to free riding, at first sight, the issue appears particularly problematic for alphabetical ordering. However, upon further examination, it becomes apparent that both schemes are subject to some degree of free riding and we show the worst case occurs under non-alphabetical ordering.

As with every theoretical study of social behavior, there is a question of what properties observed in theory apply in practice. We argue that our model offers a simple and intuitive explanation for several important phenomena observed empirically; moreover, it makes predictions that can be verified experimentally.
The study of extended models, with additional properties such as diminishing returns, larger coalitions, divisible budgets, different distributions of individual effort, and dynamic aspects of social networks
remains an important subject for future work.

\begin{table}
  \begin{tabular}{|| l | c | c | c | c | c | c | c | c | c ||}
\hline \hline
Conference & 1 & 2 & 3 & 4 & 5 & 6 & 7 & 8 & 9 \\ \hline
IJCAI 15&       5,96    & 34,86 & 29,81 & 17,88 & 6,42  & 2,75  & & &  \\ \hline
IJCAI 13&       5,69    & 29,27 & 30,31 & 20,46 & 10,10 & 3,36  & & &  \\ \hline
AAAI 15 & 7,51  & 25,26 & 25,05 & 22,96 & 12,52 & 4,17  & 1,67  & 1,04  &  \\ \hline
AAAI 14 & 3,34  & 23,62 & 30,78 & 23,62 & 10,97 & 5,72  & 1,19  & 0,71  &  \\ \hline
AAAI 13 & 5,44  & 29,20 & 34,65 & 14,35 & 8,41  & 5,94  & 1,48  & 0,49&  \\ \hline
COLT 15 & 10& 28,57     & 34,28 & 22,85 & 4,28  & & & &  \\ \hline
COLT 14 & 7,69  & 34,61 & 40,38 & 11,53 & 5,76  & & & &  \\ \hline
COLT 13 & 19,14 & 29,78 & 31,91 & 10,63 & 6,38  & 2,12  & & &  \\ \hline
NIPS 15 & 2,97  & 28,28 & 34,49 & 21,33 & 8,43  & 2,97  & 1,24  & 0,24  &  \\ \hline
NIPS 14 & 2,94  & 28,67 & 35,04 & 22,54 & 6,12  & 3,18  & 0,98  & 0,49  &  \\ \hline
NIPS 13 & 5,11  & 28,40 & 32,95 & 23,29 & 8,23  & 0,85  & 0,85& &  \\ \hline
ICML 15 & 1,85  & 30    & 33,33 & 22,22 & 9,25  & 1,85  & 0,37  & 0,74  & 0,37  \\ \hline
ICML 14 & 2,58  & 30,96 & 35,48 & 21,61 & 7,09  & 1,61  & 0,32  & 0,32  & \\ \hline
ICML 13 & 6,00  & 26,50 & 33,21 & 23,67 & 8,48  & 1,76  & 0,35 & & \\ \hline
UAI 15  & 1,01  & 51,51 & 26,26 & 12,12 & 7,07  & 1,01 & & & \\ \hline
UAI 14  & 8,33  & 34,37 & 30,20 & 19,79 & 5,20  & 2,08 & & & \\ \hline
UAI 13  & 2,73  & 39,72 & 27,39 & 23,28 & 6,84  & & & &  \\ \hline
AAMAS 15 &      0,97    & 28,57 & 37,01 & 18,18 & 10,38 & 3,57  & 1,94 & & \\ \hline
AAMAS 14 & 5,67 & 26,81 & 35,96 & 19,24 & 8,83  & 2,83  & 0,94 & & \\ \hline
AAMAS 13 & 8,26 & 29,52 & 31,49 & 20,07 & 7,87  & 2,75  & 1,18  & 0,39 &  \\ \hline \hline
\end{tabular}
\caption{Major conferences in artificial intelligence with percentage of papers written by $1, 2, \cdots, 9$ authors, respectively, during the years 2013--2015. Approximately 30\% of the published papers are completed by at most two authors.}
\end{table}

\subsection{Related Work}

Kleinberg and Oren \cite{Kleinberg_STOC} investigated a related question: Why do some academic communities over-emphasize the
importance of highly technical problems?
The authors use a noncooperative model for the allocation of scientific credit and their main finding is
that research communities may have to over-reward their key scientific challenges, to ensure that such problems are solved in a Nash
equilibrium.

The academic game is a non-transferable utility game with overlapping coalitions and is related to several types of coalitional games, such as threshold task games (\cite{Elkind}) and coalitional skill games (Bachrach and Rosenchein \cite{Bachrach}).

There are several co-authorship models in the economics and computer science literature.
de Clippel, Moulin, and Tideman \cite{Moulin} study the division of a homogeneous divisible good
when every agent reports an evaluation of the others' contribution, and establish the existence of a unique impartial and
consensual mechanism for three agents.
Jackson and Wolinsky \cite{Jackson_Wolinsky} introduce a co-author model with network externalities, where each agent has a unit
of time and can divide it among different collaborations, and study the structure of the networks in the equilibrium.
Anshelevich and Hoefer \cite{Hoefer} study the price of anarchy for contribution games on networks with concave and convex reward
functions.

\vspace{5mm}

\section{Model}
The \emph{academic game} studied here is a collaboration model defined as a network formation game. 
This simple, yet expressive formulation captures the fundamental aspect of collaboration -- namely that multiple
individuals can do more than one -- and enables the investigation of network effects due to name ordering schemes.

Let $N = \{1, \cdots, n\}$ be a set of agents. Each agent $i$ has a budget of weight $w_i$, consisting of a set of coins
$C_i = \{c_{i,1}, \cdots, c_{i, n_i}\}$. Every coin $c_{i,j}$ has a positive weight
$w_{i,j}$, and $\sum_{j=1}^{n} w_{i,j} = w_{i}$, $\forall i \in N$. The agents can work alone or in pairs to solve
different projects.
A project of weight $w$ can be solved either by:
\begin{itemize}
\item one agent who invests a coin of weight $w$ to the project, or
\item two agents, each of which contributes with one coin, such that the sum of the two coins is $w$.
\end{itemize}
An agent can participate in multiple projects simultaneously by investing a different coin in each project.
The same pair of agents can solve multiple projects together, and each coin can only be used once.

\subsection{Reward Function}

We study games with a very general class of rewards, namely convex homogenous functions. Homogeneous valuations are often used in economic theory and have been widely studied in resource allocation domains (for example, in the setting of multiple goods, the class of homogeneous valuations contains well-studied instances such as additive linear and Leontief, to constant elasticity of substitution\cite{CWG95,CGG13}). In our model, convexity translates to the property that greater effort leads to greater reward; the rate at which the increase is observed varies depending on the degree of the homogeneous function.

Formally, a function $\mathcal{F}$ is \emph{homogeneous} if there exists a positive degree $d$ such that $\mathcal{F}(t \cdot x) = t^d \cdot \mathcal{F}(x)$, for every $x$ and $t > 0$. Since we study a one-dimensional setting with convex reward, we have: $\mathcal{F}(x) =  x^d$, where $d > 1$.
We assume that for every $w \in \mathbb{R}$, there exists a project of that weight that the agent could (in principle) solve.

Finally, in each academic community there is a general perception of the significance of being the first or the second author on a paper.
Without prior knowledge about the specific paper or its authors, the relative contribution of each author on a two-authored paper
can be represented by a fixed \emph{contribution vector} $[\phi, 1-\phi]$, where $\frac{1}{2} \leq \phi < 1$. That is, the community assumes that
the contribution of the first and the second author are $\phi$\ and $(1-\phi)$\, respectively (these values can be seen as percentages of the total worth of the project). Throughout the paper, we refer to the case where $\phi = \frac{1}{2} = 1 - \phi$ as \emph{alphabetical ordering}, and to the cases where $\phi > \frac{1}{2} > 1 - \phi$ as \emph{contribution-based ordering}. Parameter $\phi$ is fixed and known to the players.  

\subsection{Coalition Structures}

Next, we define coalition structures and utility in academic games.

\begin{definition}
Given an academic game, a \emph{coalition structure} is a partition of the set of all coins, such that every coin $c_{i,j}$ of agent $i$
is either a singleton project, or is paired with a coin $c_{k,l}$ belonging to another agent $k \in N \setminus \{i\}$.
\end{definition}

\begin{definition}
Given an academic game and a coalition structure $CS$,
let $CS_i$ be the set of projects that agent $i$ contributes to, $\forall i \in N$.
The utility of $i$ is:
\[
u_i(CS) = \sum_{P_{j} \in CS_i} v_i(P_{j}),
\]
where $\{P_1, \cdots, P_m\}$ is the set of projects solved under $CS$,
$w(P_j)$ is the weight of project $P_j$, and
\[
v_i(P_{j}) =
\left\{
	\begin{array}{ll}
	w(P_j)^d & \mbox{if $i$ completes $P_j$ alone} \\
	\phi \cdot w(P_j)^d & \mbox{if $i$ is first author on $P_j$} \\
	(1 - \phi) \cdot w(P_j)^d & \mbox{if $i$ is second author on $P_j$}
	\end{array}
\right.
\]
\end{definition}

Next we give an example of an academic game. \\

\begin{example}
Consider an academic game with two agents,
where $d= 2$, agent $1$ has the set of coins $C_1 = \{c_{1,1}, c_{1,2}\}$,
agent $2$ has the set $C_2 = \{c_{2,1}\}$, and the weights of the coins are:
$w_{1,1} = 3, w_{1,2} = 1$, $w_{2,1} = 2$.
The possible coalition structures are: $CS_1 = (\{c_{1,1} \}, \{c_{1,2}\}, \{c_{2,1}\})$,
$CS_2 = (\{c_{1,1}, c_{2,1}\}, \{c_{1,2}\})$, and $CS_3 = (\{c_{2,1}, c_{1,2}\}, \{c_{1,1}\})$, where for each
project, the coins are listed by decreasing size.
The utilities of the agents are as follows:

\begin{table}[!ht]
\normalsize
  \begin{tabular}{ l  l}
$CS_1:$ & $u_1(CS_1) =  w_{1,1}^d +  w_{1,2}^d = 3^2 + 1^2 = 10$ \\
& $u_2(CS_1) = w_{2,1}^d = 2^2 = 4$.\\
& \\
$CS_2$: & $u_1(CS_2) = \phi \ \cdot (w_{1,1} + w_{2,1})^d +  w_{1,2}^d = \phi \cdot 5^2 + 1^2 = 25 \phi + 1$ \\
& $u_2(CS_2) = (1-\phi) \cdot  (w_{1,1} + w_{2,1})^d = 25 (1-\phi)$\\
& \\
$CS_3$: & $u_1(CS_3) = (1-\phi)\cdot (w_{2,1} + w_{1,2})^d + w_{1,1}^d = 9(1-\phi) + 9$ \\ & $u_2(CS_3) = \phi \cdot (w_{2,1} + w_{1,2})^d = 9 \phi$ \qed
\end{tabular}
\end{table} 
\end{example}


\section{Indivisible Budgets}
We first study the setting of indivisible budgets, where each agent owns a single coin, corresponding to the scenario where every agent is involved in a single project.
This setting is sufficient to differentiate between alphabetical and contribution ordering and highlights an interesting effect.
Namely, there exist natural settings in which alphabetical ordering encourages agents to match each others' efforts,
and as a result, it leads to the completion of larger projects.

First, we introduce pairwise stability, the standard solution concept in network formation games
\cite{Jackson}. A coalition structure $CS$ is \emph{pairwise stable} if:
\begin{itemize}
\item For all $i\in N$, $u_i(CS) \geq  w_i^d$. That is, $i$ cannot improve by allocating his coin to a singleton project.
\item For all $i, j \in N$, with $w_i \geq w_j$, either $u_i(CS) \geq \phi  \cdot (w_i + w_j)^d$ or $u_j(CS) \geq (1 - \phi)   \cdot (w_i + w_j)^d$.
That is, $i$ and $j$ cannot deviate by forming a joint project.
\end{itemize}


\subsection{Research Quality}~\label{researchQualitySection}

We show that alphabetical ordering can result in higher research quality than is possible under some contribution-based scheme.
Since agents can work either by themselves or in pairs, the most difficult project that a set of agents can solve results from the combined efforts of two of its strongest agents. We call a project of this difficulty a \emph{hard project}, which can \emph{only} be solved by two of the strongest agents, and no other combination. \\


We begin by considering identical agents.

\begin{lemma}
Consider an academic game with identical agents and indivisible budgets. Then every pairwise stable coalition structure
solves the maximum number of \emph{hard projects} whenever the credit to the first author is in the range: 
$\phi \in \left(\frac{1}{2^d}, \frac{2^d - 1}{2^d}\right)$.
%
\end{lemma}
\begin{proof}
When the agents are identical, a project is \emph{hard} if solved by two agents. Without loss of generality, we can assume that each agent has a budget of size $1$.
In order for the maximum number of hard projects to be solved in every pairwise stable equilibrium, it should be the case that
two singleton agents can strictly improve their utility by working on a joint project.
The conditions for the first and second author, respectively, are: $\phi \cdot  2^d >  1^d$ and 
$(1 - \phi)\cdot   2^d >   1^d$, or
equivalently, $\phi \in \left(\frac{1}{2^d}, \frac{2^d - 1}{2^d}\right)$.
\qed~\end{proof}

Note the maximum number of hard projects is always solved under alphabetical ordering ($\phi = \frac{1}{2}$).


Next we consider a game with \emph{heavy agents} and \emph{light agents}; the weights are normalized such that the heavy agents invest coins of weight $1$ and the light agents invest coins of weight $\lambda \in (0, 1)$.
A contribution scheme can encourage same-layer collaborations (resulting in the completion of the maximum number of hard projects), or cross-layer collaborations, or simply discourage collaboration (by giving very little credit to second authors, for example).
The proof is included in the appendix. 

\begin{theorem} \label{thm:two_layer}
Consider an academic game with indivisible budgets and two types of agents, light and heavy.
Then every pairwise stable coalition structure has:
\begin{enumerate}
	\item Only same-layer collaborations when $$\frac{(1+\lambda)^d}{2^d + (1+\lambda)^d} < \phi < \min\left(1 - \frac{1}{2^d}, \frac{1}{(1+ \lambda)^d}, \frac{2^d}{2^d + (1+\lambda)^d}\right)$$
	\item Only cross-layer collaborations when $$\max\left( 1 - \frac{1}{2^d}, \frac{1}{(1+\lambda)^d}, \frac{2^d}{2^d + (1+\lambda)^d}\right) < \phi < 1 - \left( \frac{\lambda}{1+ \lambda} \right)^d$$
	\item No collaboration when $1 - \frac{1}{2^d} < \phi < \frac{1}{(1+ \lambda)^d}$ or
	$\phi > \max\left( 1 - \frac{1}{2^d}, 1 - \left( \frac{\lambda}{1 + \lambda}\right)^d\right)$
\end{enumerate}
\end{theorem}



Observe that by setting $\phi=\frac{1}{2}$, we obtain that alphabetical ordering solves the highest number of hard projects, while ordering by contribution in the range given by Case 2 solves the highest number of intermediate projects (requiring one heavy and one light coin).

\subsection{Free Riding}\label{fairnessSection}

It has been argued that alphabetical ordering is unfair~\cite{Lake}, as it gives the same credit to all authors even when they do not contribute equally. However, the fundamental difficulty leading to free riding is that the contribution scheme is fixed, not whether it is alphabetical or contribution-based.

 Even when authors are ordering by contribution, members of academic communities have predetermined notions of the proportion of work contributed by each author\footnote{For instance, while in some communities the second author is assumed to have done moderately less than the first, in others, the contribution of the second author is considered to be negligible compared with that of the first.}. Authors cannot choose the contribution vector, as doing so would require changing the perception of the entire community.

 We show that the use of a fixed contribution vector necessarily leads to the free riding effect. In addition, the degree of free riding admitted by alphabetical ordering is \emph{not} the worst possible.

Formally, if two agents allocate weights $x$ and $y$, respectively, to a joint project, then the rewards should be proportional to the effort invested,
i.e.
$\left( \frac{x}{x + y} \right) (x+ y)^d$ and $\left( \frac{y}{x + y} \right) (x + y)^d$, respectively.

Thus, the \emph{fair contribution vector} for this project is uniquely defined as: $\mathcal{C} = \left[\frac{x}{x + y}, \frac{y}{x + y} \right]$. All other contribution vectors result in free riding.

For each agent, the \emph{free riding index} is the (normalized) difference between the perceived contribution and the actual contribution. Recall that $w(P)$ denotes the weight of project $P$.
Given a contribution scheme $\phi$, the \emph{free riding index} of agent $i$ in a coalition structure $CS$ where he solves project $P$ is:
\[ \mathcal{L}_i =
\left\{
	\begin{array}{ll}
	0 & \mbox{if $i$ completes $P$ alone} \\
	\frac{\phi \cdot w(P)-w_i}{w(P)}   & \mbox{if $i$ is the first author on $P$} \\
	\frac{(1-\phi) w(P)-w_i}{w(P)} & \mbox{if $i$ is the second author on $P$}
	\end{array}
\right.
\]

We begin by considering identical agents. Since they contribute equally to a project, alphabetical author ordering corresponds to the unique fair contribution vector for their project. In particular, larger values of $\phi$ result in more free riding for first authors.

\begin{theorem}\label{alpha fairness}
Consider an academic game with identical agents and indivisible budgets.
Then alphabetical ordering is the unique fair contribution vector, while when the credit to the first author, $\phi$, is greater than $\frac{2^d - 1}{2^d}$, then every pairwise stable coalition structure has
a free riding index of $\phi-1/2$ for at least $n/2-1$ of the agents.
\end{theorem}

\begin{proof}
It is immediate that first authors always benefit from this collaboration. The second authors would only participate as long as $(1-\phi) \cdot 2^d >  1^d$, or equivalently, $\phi < 1-\frac{1}{2^d}$. 
The free riding index of all first authors is $\phi-\frac{1}{2}$. 
\qed~\end{proof}

The free riding index is highest when the contribution vector is steep; the credit to the first author can be as high as $\phi = \frac{2^d - 1}{2^d}$ without preventing the second author from collaborating. In this case, the free riding index of all first authors is 
$\frac{1}{2}-\frac{1}{2^d}$. For example, when $d = 2$, the maximum free riding index is $25\%$. 
In general, the higher the reward of collaboration, the more free riding can occur.

\begin{corollary}\label{corr:fainessContribution}
There exist academic games and contribution-based ordering schemes such that in every pairwise stable coalition
structure, the free riding index is
$\frac{1}{2} - \frac{1}{2^d}$ for half of the agents.
\end{corollary}

Next, we find the largest free riding index that can occur under alphabetical ordering; the argument follows by the definition of the model.

\begin{lemma} \label{proposition:unfairness_loss_alphabetical_indivisible}
Consider an academic game with heavy and light agents, of weights 1 and $\lambda$ respectively.
Then in every coalition structure that is pairwise stable under alphabetical ordering, all the light agents that collaborate with heavy agents have a free riding index of $\frac{1}{2}(1-\lambda)$. The worst case is obtained when $\lambda = \frac{1}{2^{\frac{1}{d}}}-\frac{1}{2}$.
\end{lemma}

 It follows that the largest possible free riding index under alphabetical ordering is no greater than $2^{-\frac{1}{d}}-\frac{1}{2}$.
For example, when $d=2$, the free riding index is $\approx 20$\%.

While in the previous results we showed that the contribution vector can affect as many as half of the agents involved, the same worst case bounds hold for individual agents in arbitrary games. 


\begin{theorem}
Consider an academic game in which the agents have indivisible budgets of arbitrary sizes. Then the highest free riding index of any agent occurs under contribution-based ordering.
Moreover, the highest amount of free-riding that occurs in any project solved under alphabetical ordering is smaller than under contribution-based ordering.
\end{theorem}
\begin{proof}
Note that the bound in Lemma~\ref{proposition:unfairness_loss_alphabetical_indivisible} represents the highest free riding index that any agent can incur under alphabetical ordering, namely $\frac{1}{2}^{\frac{1}{d}} - \frac{1}{2}$. The result follows from Lemma~\ref{proposition:unfairness_loss_alphabetical_indivisible} and Corollary~\ref{corr:fainessContribution}, which imply that the amount of free riding can be as high as $\frac{1}{2}-\frac{1}{2^d}$ under some contribution-based ordering schemes.

To show that the highest amount of free-riding in any project solved under alphabetical ordering is smaller than under contribution-based ordering, we need to show that $\frac{1}{2}^{\frac{1}{d}} - \frac{1}{2} \leq \frac{1}{2} - \frac{1}{2^d}$, or equivalently, $\frac{1}{2}^{d} + \frac{1}{2}^{\frac{1}{d}} \leq 1$.
The inequality can be shown by analyzing the behavior of the following function: $f: \mathbb[1, \infty) \rightarrow \mathbb{R}$, where $f(x) = \frac{1}{2}^{x} + \frac{1}{2}^{\frac{1}{x}}$. Note that $f(1) = 1$ and $\lim_{x \to \infty} f(x) = 1$. The first derivative of $f$ is 
$$f'(x) = \dfrac{\ln\left(2\right)}{x^2{\cdot}2^\frac{1}{x}}-\dfrac{\ln\left(2\right)}{2^x}.$$ The function $f'(x) = 0$ has a unique solution\footnote{This can be easily checked; see, e.g., http://www.wolframalpha.com} at the point $x_0 > 1$ with the property that $2^{x_{0}} = x_0^2 \cdot 2^{\frac{1}{x_0}}$, and moreover $f'(x) < 0$ for all $x \in (1,x_0)$ and $f'(x) > 0$ for all $x > x_0$. Then $f(x) < 1$ for all $x>1$ as required.
\qed~\end{proof}


 To conclude, even though at first sight, alphabetical ordering appears to be particularly susceptible to free-riding, further investigation reveals that
 no contribution scheme is immune to free riding; neither listing authors alphabetically nor listing them in decreasing order of contribution can eliminate this effect. Moreover, the worse case attained using alphabetical ordering is better than that of some contribution based schemes.

\subsection{Stability} \label{subsec:stable}

In this section we show that alphabetical author ordering guarantees the existence of a pairwise stable coalition structure; furthermore, it can always be found in polynomial time, while non-alphabetical ordering may not have any stable coalition structure.

In the following theorems, a tie-breaking rule is defined as usual (i.e. lexicographic), and will mean the function that decides which of two authors with equal contributions comes first. 

\begin{theorem}
Every academic game with identical agents and indivisible budgets has a pairwise stable coalition structure.
\end{theorem}
\begin{proof}
Without loss of generality, assume that the deterministic tie-breaking rule lists the
agents in the order $[1, \cdots, n]$.
If $\phi \geq 1 - \frac{1}{2^d}$, then the singleton coalition structure, $CS = (\{1\}, \cdots, \{n\})$ is pairwise stable,
since no agent can improve their utility by being second author on a project.
If $\phi < 1 - \frac{1}{2^d}$, there are two cases. If
$n$ is even, then $CS^{'} = (\{1,2\}, \cdots, \{n-1, n\})$ is pairwise stable.
If $n$ is odd, then $CS^{''} = (\{2,3\}, \cdots, \{n-1, n\}, \{1\})$, is pairwise stable, since no agent has an incentive to
switch to a singleton, the coalition structure given by $CS^{''} \setminus \{1\}$ is pairwise stable (by case 1), and
agent $1$ cannot be part of a deviating pair, since no other agent has an incentive to join $1$ as a second author.
\qed~\end{proof}

\begin{theorem}
Consider an academic game with different agents and indivisible budgets.
Under alphabetical ordering, a pairwise stable coalition structure is guaranteed to exist and can be
computed in polynomial time.
\end{theorem}
\begin{proof}
Let $N = \{1, \cdots, n\}$ be the set of agents and without loss of generality, let $w_1 \geq w_2 \geq \cdots \geq w_n$.
Start with an empty coalition structure: $CS = \emptyset$. Iteratively, pair whenever possible the two agents with the heaviest weights
among the remaining agents. Let $\{k, k+1, \cdots, n\}$ be the remaining set of agents during some iteration.
If 
\[
\frac{1}{2} \cdot (w_{k} + w_{k+1})^{d} \geq w_{k}^{d}
\]
and 
\[
\frac{1}{2}\cdot (w_{k}+ w_{k+1})^{d} \geq w_{k+1}^{d}
\]
then
let $CS \leftarrow CS \cup \{k, k + 1\}$, otherwise, $CS \leftarrow CS \cup \{k\}$.
We claim that the resulting coalition structure, $CS$, is pairwise stable. If $CS$ contains coalition $\{1, 2\}$,
then agent $1$ does not have an incentive to form another
pair or move to a singleton, since agent $1$'s utility in $CS$, $u_1(CS)$, verifies the following inequalities:
\[
u_1(CS) \geq w_{1}^{d}
\]
and 
\[
u_1(CS) \geq \frac{1}{2} \cdot (w_1 + w_j)^{d}, \forall j \in N \setminus \{1\}
\]
Similarly, agent $1$ does not deviate if it is in a singleton coalition structure. Iteratively, whenever the first $k$ agents do
not have an incentive to deviate, agent $k + 1$ does not have an incentive to deviate either. Thus $CS$ is pairwise stable.
\qed~\end{proof}

On the other hand, contribution-based ordering
does not guarantee the existence of stable coalition structures even under fixed tie-breaking rules.

\begin{theorem}
There exist academic games with different agents and indivisible budgets,
such that when contribution-based ordering is used, no pairwise stable coalition structure exists.
\end{theorem}
\begin{proof}
Consider a three agent game, with weights $1$, $1+ \varepsilon$, and $1 + 2 \varepsilon$, respectively, where $\alpha = 1$, $d= 2$, $\varepsilon = 0.8$,
and $\phi = 0.8$. It can be easily verified that no coalition structure is stable. The singleton coalition structure is blocked by the agents
with weights $\{1,1+\varepsilon\}$, the coalition structure $CS = (\{1+\varepsilon, 1\}, \{1 + 2 \varepsilon\})$ is blocked by $\{1+2\varepsilon,1\}$,
$CS^{'} = (\{1+2\varepsilon,1+\varepsilon\},\{1\})$ is blocked by $\{1+\varepsilon,1\}$, and
$CS^{''} = (\{1+2\varepsilon,1\}, \{1+\varepsilon\})$ is blocked by $\{1+2\varepsilon,1+\varepsilon\}$.
\qed~\end{proof}

Next we show that pairwise stable coalition structures can be found in polynomial time. The proof uses a connection with the stable roommates problem, which is a generalization of the stable marriage problem and can be roughly stated as follows.
\begin{quote}
\emph{\textbf{Stable Roommates Problem}: There are $2n$ participants and $n$ rooms (each accommodating two people); each participant wishes to find a roommate to live with. The participants have a ranking for the others in strict order of preference. 
A matching is a set of $n$ pairs of participants, such that the agents in each pair will share a room together. The goal is to find stable matching, i.e. where it will not be the case that for two participants $i$ and $j$, they each prefer each other to their current roommates.}
\end{quote}
The stable roommates problem was studied in a very influential paper by Irving~\cite{Irving85}; a polynomial time algorithm for preferences with ties (where participants may have equal value for some roommates) and incomplete lists (where a participant prefers living alone to living with some of the potential roommates) is presented for instance in the Ph.d. thesis of Scott~\cite{Scott}. \\

We now rely on the roommates problem to prove the following result: 

\begin{theorem}
Consider an academic game with different agents and indivisible budgets. Then a pairwise stable coalition structure can be found in polynomial time when it exists.
\end{theorem}
\begin{proof}
The game is an instance of the
stable roommates problem with ties and incomplete lists, where an agent $i$ finds another agent $j$ unacceptable if $i$ prefers working alone instead
of forming a pair with $j$.  In our setting, each agent can be mapped to a ``roommate'' and the preferences are induced by how valuable the coin of each agent is.
\qed~\end{proof}

\section{Discrete Budgets}

We now turn our attention to the general model, where each agent has multiple coins, allowing agents to work on multiple projects simultaneously.
The model of discrete budgets allows uncovering several phenomena that cannot be observed
in the indivisible budget setting, such as the following:
there exist many games in which the \emph{contribution vector does not matter}, since the agents can perform rotations, by alternating between being first and second author on joint projects. Rotations can allow agents to reach optimal research quality as well
as obtain perfect fairness.

Our solution concept is pairwise stability for games with overlapping coalition structures.
Given that an agent can be involved in multiple projects simultaneously, it is important that
one estimates correctly the reactions from the rest of the agents before agreeing to participate in a deviation. We follow the recent literature on overlapping coalition formation games
(\cite{Elkind}, Zick, Chalkiadakis, and Elkind \cite{Zick}), and study \emph{sensitive reactions} to a deviation.
In short, when agent $i$ is involved in a deviation from a coalition structure $CS$, $i$ can expect that:
\begin{itemize}
\item Every non-deviating agent who is hurt by the deviation retaliates
and drops all the projects with $i$. Note that unless $i$ and $j$ agreed to deviate together,
an agent $j$ is hurt by the deviation when at least one of $j$'s projects has been discontinued by the deviator(s).
\item The unaffected agents are neutral and maintain all of their existing projects with $i$.
\end{itemize}

\begin{definition}
A coalition structure $CS$ is pairwise stable if:
\begin{itemize}
\item No agent $i$ can drop some of his existing projects and strictly improve in the new
coalition structure, $CS^{'}$
\item No two agents $i$ and $j$, can rearrange the projects among themselves and possibly drop some of the
projects with the remaining agents, such that both $i$ and $j$ strictly improve their utility in $CS^{'}$,
\end{itemize}
where $CS^{'}$ is the resulting coalition structure, in which non-deviators have sensitive reactions to a deviation.
\end{definition}

The next example illustrates pairwise stability for discrete budgets.

\begin{example}
Consider an academic game with three agents, sets of coins: $C_1 = \{c_{1,1}, c_{1,2}\}$, $C_2 = \{c_{2,1}, c_{2,2}\}$, and
$C_3 = \{c_{3,1}, c_{3,2}\}$, and the coalition structure $CS = (\{c_{1,1}$, $c_{2,1}\}$, $\{c_{2,2}$, $c_{3,1}\}$, $\{c_{1,2}$, $c_{3,2}\})$.

If agent $1$ deviates by allocating the coin $c_{1,2}$ to a singleton project, then $1$ expects that the resulting coalition
structure is $CS^{'}$ $=$ $(\{c_{1,1}$, $c_{2,1}\}$, $\{c_{1,2}\}$, $\{c_{3,2}\}$, $\{c_{2,2}$, $c_{3,1}\})$, since agent $2$ is not hurt by the
deviation.

On the other hand, if the deviating coalition is $\{1,2\}$ and the deviation consists of forming the joint project $\{c_{1,1}, c_{2,2}\}$,
then $1$ and $2$ expect the resulting coalition structure is $CS^{''}$ $=$ $(\{c_{1,1}$, $c_{2,2}$ $\}$, $\{c_{2,1} \},$ $\{c_{1,2}\},$
$\{c_{3,1}\}$, $\{c_{3,2}\})$, since agent $3$ is hurt by the deviation and drops \emph{all} the projects with the deviators.
\qed~\end{example}

\subsection{Rotations}

Agents can sometimes overcome the limitations of a fixed contribution scheme. That is, they can simultaneously solve the highest number of hardest projects and eliminate free riding, regardless of the contribution vector.
We refer to this phenomenon as \emph{rotations}:
agents collaborating on multiple projects agree that one of them is the first author on half of their projects, while the other is first on the remaining projects, regardless of whether this represents their actual contributions. \\

We first demonstrate this for budgets with multiple identical coins.

\begin{theorem}\label{rotations1}
There exist academic games with discrete budgets and multiple identical coins
such that for every $\phi$, the maximum number of hard projects is solved in a pairwise stable equilibrium and no free riding occurs.
\end{theorem}
\begin{proof}
Let $\phi < 1$ and consider a two agent game, such that agent $1$ has coins $\{c_{1,1}, c_{1,2}\}$, agent $2$
has coins $\{c_{2,1}, c_{2,2}\}$, and all the coins have weight $1$. Consider the coalition structure $CS = (C_1, C_2)$, where
$C_1 = \{c_{1,1}, c_{2,1}\}$ and $C_2 = \{c_{2,2}, c_{1,2}\}$, such that agent $1$ is the first author on project $C_1$ and agent
$2$ is the first author on project $C_2$. 

It can be verified that both agents receive the best possible utility, which coincides
with the fair allocation given by alphabetical ordering.
Moreover, $CS$ is pairwise stable:
no agent can gain by investing the coin from their second-author project to a singleton project, since the other agent
will retaliate and drop the other project as well.
\qed~\end{proof}


Rotations can also be used to eliminate free riding when coins are not required to be identical, and projects require the combination of different coin types.\\


Next we study rotations in games where there exists a \emph{conference tier}, i.e. for which there is a threshold $t > 0$ such that a project of weight $w$ gives reward $\mathcal{F}(w)$ if $w \geq t$, and zero otherwise.

\begin{theorem}\label{rotations2}
Let $1+\lambda$ be the conference tier. Then there are academic games with discrete budgets, where each agent has an equal number of light and heavy coins, of weight $1$ and $\lambda \in (0,1)$, respectively, such that there exist pairwise stable coalition structures with no free riding.
\end{theorem}
\begin{proof}
Let $\phi < 1$ and consider a two agent game, such that agent $1$ has coins $\{c_{1,1}, c_{1,2}\}$, agent $2$
has coins $\{c_{2,1}, c_{2,2}\}$, and all the coins have weight $1$. Consider the coalition structure $CS = (C_1, C_2)$, where
$C_1 = \{c_{1,1}, c_{2,1}\}$ and $C_2 = \{c_{2,2}, c_{1,2}\}$, such that agent $1$ is the first author on project $C_1$ and agent
$2$ is the first author on project $C_2$. It can be verified that both agents receive the best possible utility, which coincides
with the fair allocation given by alphabetical ordering.
Moreover, the coalition structure is pairwise stable.
None of the agents can gain by investing the coin from their second-author project to a singleton project, since the other agent
will retaliate and drop the other project as well.
\qed~\end{proof}

\subsection{Implications for the Social Network}
For the next result we illustrate the following phenomenon observed in~\cite{Newman}: when the agents use ordering by
contribution, they have more co-authors than when using alphabetical ordering.
We consider the setting in which every agent has a budget consisting of heavy coins and light coins.
The light coins represent very little effort, such as ``cheap talk'',
but can contribute to improving the quality of a paper. Allowing for cheap talk results in a much higher number of collaborations.

\begin{theorem}~\label{social network result}
Consider an academic game with discrete budgets, where each agent has several heavy and light coins,
of sizes $1$ and $\varepsilon$, respectively, such that $0 < \varepsilon \ll 1$. Moreover, the conference tier is $1$ and
each agent has more heavy coins than light coins.

Then whenever $\phi > \max \left(\frac{2^d}{2^d + (1+\varepsilon)^d}, \frac{1}{(1+\varepsilon)^d}\right)$,
every pairwise stable equilibrium solves the maximum number of projects and the average number of
collaborators per agent is the highest possible.
\end{theorem}
\begin{proof}
To ensure that every pairwise stable coalition structure solves the maximum number of collaborations, the best investment of a
heavy coin should be to pair it with a small coin.
That is, an agent prefers being the first author on a project of weight $1+\varepsilon$ instead of either
second author on a hard project or the only author
on a singleton project. The conditions are: $(1-\phi) \cdot 2^d < \phi \cdot (1+\varepsilon)^d$ and $\phi \cdot (1+\varepsilon)^d > 1$; equivalently, $\phi$ $>$ $\max \left(\frac{1}{(1+\varepsilon)^d}, \frac{2^d}{2^d + (1+\varepsilon)^d}\right)$.

Then in every pairwise stable coalition structure, all the heavy coins are paired with small coins, and the average number
of collaborators is maximal.
Note that while alphabetical ordering solves the highest number of hard projects, both the number of projects
completed and the number of collaborators are twice as low.
Finally, the singleton coalition structure solves the same number of projects above the conference tier.
However, in this case, the agents have no collaborators, and the quality of the projects is lower compared to the case when cheap talk is allowed.
\qed~\end{proof}

There exist games in which ordering by contribution allows to simultaneously maximize the number of hard projects
and the total number of projects. We illustrate this phenomenon when there exist both an upper and lower bound on the
hardness of the rewarded projects.

\begin{corollary}
Consider an academic game with discrete budgets, where each agent has several heavy and light coins, of sizes $1$ and $\lambda \in (0,1)$,
respectively, the conference tier is $1$, and the maximum project hardness is $1 + \lambda$.
Each agent has more heavy coins than small coins.

Then whenever $\phi > \max \left(\frac{2^d}{2^d + (1+\lambda)^d}, \frac{1}{(1+\lambda)^d}\right)$,
every pairwise stable equilibrium solves the maximum number of projects, each of the projects solved is the hardest possible,
and the average number of collaborators per agent is the highest possible.
\end{corollary}

\section{Discussion}


We introduced a basic game theoretic model for studying author ordering schemes, which already illustrates interesting phenomena that can occur in richer domains. 
The model offers a compelling explanation for several real-world phenomena, showing that alphabetical ordering is positively correlated with research quality in some scenarios, while contribution based ordering results in a larger number of projects completed and denser social networks.

Our model makes several predictions on the effects of author ordering schemes, which prompt further theoretical and empirical study. 
In particular, we show that rotations can be used to overcome the limitations of fixed ordering schemes and that the worst case of free riding occurs under some contribution-based schemes. It would be interesting to empirically evaluate how frequently free riding and rotations appear in practice, as well as their influence on individuals and their communities. 

In this paper, we focus on modeling communities where the contribution scheme is established. There are other communities, however, where there is some ambiguity about the contribution schemes (for example, the dominant scheme may be contribution-based, but some authors may still choose to order alphabetically). To further investigate this phenomenon, it would be interesting to study the impact of a probabilistic contribution scheme.  Another natural generalization is to increase the number of authors who may collaborate on the same project. This generalization may lead to additional insight on the differences between contribution and alphabetical schemes with regards to research quality, free riding, and density of the collaboration network. 

Another important direction for future work is to understand precisely the conditions under which alphabetical ordering is better suited than contribution-based ordering, and vice-versa. In those communities where alphabetical ordering is indeed the closest to optimal scheme, policy changes may be called for to alleviate the effect of alphabetical ordering of unfairly favoring the authors with earlier names in the alphabet. Such changes could include changing the citation styles for alphabetical papers or possibly using random ordering, with a note that the contribution is meant to be weighted equally.
Moreover, it would be interesting to study models that capture additional realistic phenomena, such as reward functions with diminishing returns, arbitrary coalition sizes, heterogenous skill sets, and dynamics of social networks that may influence the equilibria reached.

More generally, the following implementation theory question remains open: \emph{Given a scientific community, what is the optimal credit allocation scheme?}
This work is a first step in the direction of understanding this question, which is at the heart of resource allocation in academic research and arguably the long term development of society.

\begin{acknowledgements}
We would like to thank Peter Bro Miltersen, Yair Zick, and the anonymous reviewers for useful feedback that helped improve the paper.
Simina Br\^{a}nzei acknowledges support from the Danish National Research Foundation
and The National Science Foundation of China (under the grant 61361136003) for
the Sino-Danish Center for the Theory of Interactive Computation and from the Center for
Research in Foundations of Electronic Markets (CFEM), supported by the Danish
Strategic Research Council. Simina was also supported by ISF grant 1435/14 administered by the Israeli Academy of Sciences and Israel-USA Bi-national Science Foundation (BSF) grant  2014389, and the I-CORE Program of the Planning and Budgeting Committee and The Israel Science Foundation.
\end{acknowledgements}

\bibliographystyle{spmpsci}      

\newpage

\section{Appendix}

In this section we include the proof of Theorem \ref{thm:two_layer}. \\

\textbf{Theorem \ref{thm:two_layer}} (restated):
\emph{
Consider an academic game with indivisible budgets and two types of players, heavy and light (with weights $1$ and $\lambda$, respectively), where $0 < \lambda < 1$.
Then every pairwise stable coalition structure has
\begin{enumerate}
	\item Only same-layer collaborations when $$\frac{(1+\lambda)^d}{2^d + (1+\lambda)^d} < \phi < \min\left(1 - \frac{1}{2^d}, \frac{1}{(1+ \lambda)^d}, \frac{2^d}{2^d + (1+\lambda)^d}\right)$$
	\item Only cross-layer collaborations when $$\max\left( 1 - \frac{1}{2^d}, \frac{1}{(1+\lambda)^d}, \frac{2^d}{2^d + (1+\lambda)^d}\right) < \phi < 1 - \left( \frac{\lambda}{1+ \lambda} \right)^d$$
	\item No collaboration when $1 - \frac{1}{2^d} < \phi < \frac{1}{(1+ \lambda)^d}$ or
	$\phi > \max\left( 1 - \frac{1}{2^d}, 1 - \left( \frac{\lambda}{1 + \lambda}\right)^d\right)$.
\end{enumerate}
}

\begin{proof}
The proof follows from Lemma \ref{lem:two_player_same_layer}, Lemma \ref{lem:two_layer_cross}, and Lemma \ref{lem:two_layer_none} below; each Lemma covers one of the cases in the theorem.
\qed~\end{proof}

\begin{lemma} \label{lem:two_player_same_layer}
Consider an academic game with indivisible budgets and two types of players, of weights $1$ and $\lambda$, respectively, where $0 < \lambda < 1$.
Then every pairwise stable coalition structure has only same-layer collaborations when
 $\frac{(1+\lambda)^d}{2^d + (1+\lambda)^d} < \phi < \min\left(1 - \frac{1}{2^d}, \frac{1}{(1+ \lambda)^d}, \frac{2^d}{2^d + (1+\lambda)^d}\right)$.
\end{lemma}
\begin{proof}
For every pairwise stable coalition structure to solve the maximum number of same-layer collaborations, it should be the case
that whenever a coalition structure contains:
\begin{itemize}
\item 2 identical singleton projects: the two players can improve their utility by deviating to a pair. That is, $(1-\phi)\cdot 2^d > 1$
and $(1 - \phi) \cdot (2\lambda)^d > \lambda^d$, i.e. $\phi < 1 - \frac{1}{2^d}$.
\item 2 cross-layer projects: then there exists an improving deviation by two players from the same layer.
It is sufficient to require that the two heavy players involved in the cross layer projects deviate together: $\phi \cdot 2^d > \phi \cdot (1 + \lambda)^d$ and
$(1-\phi) \cdot 2^d > \phi \cdot (1 + \lambda)^d$. Thus $\phi < \frac{2^d}{2^d + (1 + \lambda)^d}$.
\item 1 cross-layer project and 1 heavy singleton project: then the two heavy players deviate to a pair. The two heavy players can deviate
to a pair when $\phi \cdot 2^d > \phi \cdot (1 + \lambda)^d$ and $(1-\phi)2^d > 1$, i.e. $\phi < 1 - \frac{1}{2^d}$.
\item 1 cross-layer project and 1 light singleton project: then the two light players deviate to a pair. The two light players
can deviate to a pair when $\phi \cdot (2\lambda)^d > (1-\phi)(1+\lambda)^d$ and $(1-\phi)(2\lambda)^d > \lambda^d$.
That is, $\frac{(1+\lambda)^d}{(2\lambda)^d + (1+\lambda)^d} < \phi < 1 - \frac{1}{2^d}$.
\end{itemize}
In addition, we require that no intermediate project is solved, even when the maximum number of hard projects is completed.
That is, a coalition of weight $1+ \lambda$ is blocked by a deviation to a singleton by one of the players, i.e. $\phi < \frac{1}{(1+\lambda)^d}$ or $\phi > 1 - \left( \frac{\lambda}{1 + \lambda} \right)^d$.
It follows that
 $\frac{(1+\lambda)^d}{2^d + (1+\lambda)^d} < \phi < \min\left(1 - \frac{1}{2^d}, \frac{1}{(1+ \lambda)^d}, \frac{2^d}{2^d + (1+\lambda)^d}\right)$.
\qed~\end{proof}

\begin{lemma}  \label{lem:two_layer_cross}
Consider an academic game with indivisible budgets and two types of players, of weights $1$ and $\lambda$, respectively, where $0 < \lambda < 1$.
Then every pairwise stable coalition structure has the maximum number of cross-layer collaborations when
$\max\left(\frac{1}{(1+\lambda)^d}, 1 - \frac{1}{2^d}, \frac{2^d}{2^d+ (1+\lambda)^d}\right) < \phi < 1 - \left( \frac{\lambda}{1+ \lambda} \right)^d$.
\end{lemma}
\begin{proof}
For every pairwise stable coalition structure to solve the maximum number of cross-layer collaborations, it should be the case that
whenever a coalition structure contains:
\begin{itemize}
\item 2 singleton projects of different weights: then the two players can deviate to a pair. That is, $\phi \cdot (1 + \lambda)^d > 1$
and $(1-\phi)(1 + \lambda)^d > \lambda^d$, i.e. $\frac{1}{(1+\lambda)^d} < \phi < 1 - \left(\frac{\lambda}{1+\lambda}\right)^d$.
\item 1 same-layer project by two heavy players and 1 light singleton project: then the light player deviates with one of the two heavy players.
That is,
$\phi \cdot (1 + \lambda)^d > (1 - \phi) \cdot 2^d$
and
$(1-\phi)(1+\lambda)^d > \lambda^d$, and so $\frac{1}{1 + \left( \frac{1+\lambda}{2} \right)^d} < \phi < 1 - \left( \frac{\lambda}{1 + \lambda} \right)^d$.
\item 1 same-layer project by two light players and 1 heavy singleton project: then the heavy player deviates with one of the light players.
That is, $(1-\phi)\cdot (1 + \lambda)^d > (1-\phi)(2\lambda)^d$ and $\phi \cdot (1 + \lambda)^d > 1$, i.e. $\phi > \frac{1}{(1 + \lambda)^d}$.
\item 2 same-layer projects of different weights: Then one of the heavy players deviates with one of the light players.
That is, $\phi \cdot (1 + \lambda)^d > (1 - \phi) \cdot 2^d$ and
$(1-\phi)\cdot (2\lambda)^d < (1-\phi) \cdot (1 + \lambda)^d$, i.e. $\phi > \frac{1}{1 + \left( \frac{1+\lambda}{2}\right)^d}$.
\end{itemize}
In addition, we require that even when the maximum number of cross layer collaborations occurs, if there exists a same-layer project, it is blocked by a deviation of the second player, who prefers to work alone. That is, $\phi > 1 - \frac{1}{2^d}$, and so
$\max\left(\frac{1}{(1+\lambda)^d}, 1 - \frac{1}{2^d}, \frac{2^d}{2^d + (1+\lambda)^d}\right) < \phi < 1 - \left( \frac{\lambda}{1+ \lambda} \right)^d$.
\qed~\end{proof}

\begin{lemma}  \label{lem:two_layer_none}
Consider an academic game with indivisible budgets and two types of players, of weights $1$ and $\lambda$, respectively, where $0 < \lambda < 1$.
Then no pairwise stable coalition structure has collaborations when $1 - \frac{1}{2^d} < \phi < \frac{1}{(1+ \lambda)^d}$ or
$\phi > \max\left( 1 - \frac{1}{2^d}, 1 - \left( \frac{\lambda}{1 + \lambda}\right)^d\right)$.
\end{lemma}
\begin{proof}
There are two cases. If the coalition structure contains a same-layer project, then it is blocked by one of the players, who deviates to a singleton.
That is, $(1-\phi)\cdot 2^d < 1$. If the coalition structure contains a cross-layer project, then again one of the players deviates to a singleton.
That is,
either $\phi \cdot (1 + \lambda)^d < 1$ or $(1-\phi) \cdot (1 + \lambda)^d < \lambda^d$.
Thus every pairwise stable coalition structure contains only singleton projects when either $1 - \frac{1}{2^d} < \phi < \frac{1}{(1+ \lambda)^d}$ or
$\phi > \max\left( 1 - \frac{1}{2^d}, 1 - \left( \frac{\lambda}{1 + \lambda}\right)^d\right)$.
\qed~\end{proof}

\end{document}